\documentclass[letterpaper, 10 pt, conference]{ieeeconf}  

\IEEEoverridecommandlockouts                              
\makeatletter

\let\proof\@undefined
\let\endproof\@undefined

\makeatother

\usepackage{graphicx}
\usepackage{caption}
\usepackage{subcaption}
\usepackage{enumerate}
\usepackage{hyperref}
\usepackage{upgreek}
\usepackage{array}
\usepackage{multirow}
\usepackage{amsmath} 
\usepackage{amssymb}  
\usepackage{amsthm}  
\usepackage{bm}
\usepackage{lipsum}
\usepackage[linesnumbered, ruled]{algorithm2e}
\usepackage{color}
\usepackage{cite}

\newtheorem{proposition}{Proposition}

\newtheorem{assumption}{Assumption}
\newtheorem{lemma}{Lemma}
\newtheorem{theorem}{Theorem}

\newtheorem{cor}{Corollary}

\title{\LARGE \bf
Minimax Control of Ambiguous Linear Stochastic Systems\\ Using the Wasserstein Metric
\thanks{This work was supported in part by the National Research Foundation of Korea funded by the MSIT(2020R1C1C1009766), the Creative-Pioneering Researchers Program through SNU, and Samsung Electronics.}
}

\author{Kihyun Kim \and Insoon Yang
\thanks{K. Kim, and I. Yang are with the Department of Electrical and Computer Engineering and ASRI, Seoul National University, Seoul, 08826, Korea {\tt\small \{hahakhkim, insoonyang\}@snu.ac.kr}}%
}

\begin{document}

\maketitle
\thispagestyle{empty}
\pagestyle{empty}


\begin{abstract}
In this paper, we propose a minimax linear-quadratic control method  to address the issue of inaccurate distribution information in practical stochastic systems. 
To construct a control policy that is robust against errors in an empirical distribution of uncertainty, our method is to adopt an adversary, which selects the worst-case distribution. 
To systematically adjust the conservativeness of our method,
the opponent receives a penalty proportional to  the amount, measured with the Wasserstein metric, of deviation from the empirical distribution. 
In the finite-horizon case, using a Riccati equation,
we derive a closed-form expression of the unique optimal policy and the opponent's policy that generates the worst-case distribution.
This result is then extended to the infinite-horizon setting by identifying conditions under which  the Riccati recursion converges to the unique positive semi-definite solution to an associated algebraic Riccati equation (ARE).
The resulting optimal policy is shown to stabilize the expected value of the system state under the worst-case distribution. 
We also discuss that our method can be interpreted as a distributional generalization of the $H_\infty$-method.
\end{abstract}


\section{Introduction}

Ambiguity, or uncertainty about uncertainty, in stochastic systems is one of the most fundamental challenges in the practical implementation of stochastic optimal controllers~\cite{Petersen2000, Tzortzis2015}. 
The true probability distribution of underlying uncertainty is not known in ambiguous stochastic systems. 
We often only have an access to samples generated according to the distribution. 
In practice, estimating an accurate distribution from such observations is technically challenging  due to insufficient data and inaccurate statistical models, among others. 
Using inaccurate distribution information in the construction of an optimal policy may significantly decrease the control performance~\cite{Nilim2005, Yang2019} and can even cause unwanted system behaviors, in particular, violating safety constraints~\cite{Yang2018aut}.
The focus of this work is to develop and analyze a discrete-time minimax control method that is robust against uncertainties or errors in such distribution information. 

Our method is closely related with the literature in distributionally robust control (DRC).
DRC methods seek a control policy that minimizes an expected cost of interest under the worst-case distribution in a so-called \emph{ambiguity set}.
Several  types of ambiguity sets have been employed in DRC
 using moment constraints~\cite{Xu2012, VanParys2016}, confidence sets~\cite{Yang2017cdc}, relative entropy~\cite{Petersen2000, Ugrinovskii2002}, 
total variation distance~\cite{Tzortzis2015, Tzortzis2016}, and Wasserstein distance~\cite{Yang2017lcss, Yang2018}.
Such choices of ambiguity sets have largely been motivated by the literature in distributionally robust optimization (DRO)~\cite{Delage2010, BenTal2013, Wiesemann2014, Esfahani2015, Zhao2018, Gao2016}. 
In particular, DRO and DRC with the Wasserstein ambiguity set possess salient features such as a probabilistic out-of-sample performance guarantee and computational tractability~\cite{Esfahani2015, Zhao2018, Gao2016, Blanchet2018, Yang2018}.

In this paper, we propose a minimax linear-quadratic control method for ambiguous stochastic systems, inspired by Wasserstein DRC. 
Instead of using an ambiguity set, our method pursues distributional robustness by adopting a penalty term in the objective function. 
Specifically, a hypothetical opponent selects the worst-case distribution to maximize a quadratic cost of interest, while the controller aims to minimize the same cost.
To limit the conservativeness of the resulting control policy, we penalize the opponent by the amount, measured with the Wasserstein metric, of deviating from an empirical distribution. 
Our method can be interpreted as a Lagrangian relaxation of Wasserstein DRC. 

In the finite-horizon case, we derive a Riccati equation and a closed-form expression of the unique optimal policy, which is linear, and the opponent's policy that generates the worst-case distribution. This result  also confirms that our minimax method generalizes the standard linear-quadratic-Gaussian (LQG) method in the sense that the distributional robustness of our optimal policy is tunable. 
In the infinite-horizon setting,  we identify a condition under which the solution to the Riccati equation converges to a symmetric positive semi-definite (PSD) solution to an algebraic Riccati equation (ARE). 
By taking a generalized eigenvalue approach,  our result is strengthened so that the converged solution corresponds to a unique symmetric PSD solution to the ARE under an additional observability condition. 
The corresponding optimal policy is shown to stabilize the expected value of the system state under the worst-case distribution. 
Furthermore, we establish theoretical connections to the classical $H_\infty$-method. Interestingly, our  method can be understood as a distributional generalization of the $H_\infty$-method, thereby bridging the gap between stochastic and robust control.
The effectiveness of our method is analyzed through a power system frequency control problem. 

The remainder of this paper is organized as follows. 
In Section~\ref{sec:setup}, we introduce the minimax linear-quadratic control problem using Wasserstein distance. 
Section~\ref{sec:finite} is devoted to the finite-horizon case. 
In Section~\ref{sec:infinite}, we present several results that connect the finite-horizon and infinite-horizon cases and discuss closed-loop stability.  In Section~\ref{sec:robust}, we identify relations between our method and the $H_\infty$-method. 
Section~\ref{sec:exp} presents the results of our numerical experiments. 

\section{The Setup}\label{sec:setup}

\subsection{Ambiguity in Stochastic Systems}

Consider a linear discrete-time system of the form
\begin{equation} \label{sys}
    x_{t+1} = Ax_t + Bu_t + \Xi w_t,\quad t \geq 0,
\end{equation}
where $x_t \in \mathbb{R}^n$ and $u_t \in \mathbb{R}^m$ denote the system state and input, respectively. 
Here, $w_t \in  \mathbb{R} ^k$ is a random disturbance vector with probability distribution $\mu_t \in \mathcal{P}(\mathbb{R}^k)$, where $\mathcal{P}(\mathbb{W})$ denotes the set of Borel probability measures on $\mathbb{W}$.
In addition,
$A\in\mathbb{R}^{n \times n}$, $B\in\mathbb{R}^{n \times m}$, and $\Xi\in\mathbb{R}^{n \times k}$ are time-invariant system matrices.

In practice, it is challenging to obtain the true probability distribution $\mu_t$ of $w_t$. 
One of the most straightforward ways to estimate the distribution is to construct the following empirical distribution from  sample data $\{ \hat{w}^{(1)}_t, \ldots, \hat{w}^{(N)}_t\}$ of $w_t$:
\[
\nu_t := \frac{1}{N} \sum_{i=1}^N \delta_{\hat{w}^{(i)}_t}.
\]
However, it is not desirable to use this empirical distribution in controller design because the control performance would deteriorate while the true distribution deviates from $\nu_t$.

\subsection{Minimax Stochastic Control with Wasserstein Distance}

Let $\pi := (\pi_0, \pi_1, \ldots)$ denote a deterministic Markov control policy, where $\pi_t$ maps the current state $x_t$ to 
an input $u_t$.\footnote{For ease of exposition, we focus on deterministic Markov policies. However, all the results in this paper are valid even when considering randomized history-dependent policies for both players by the optimality result in~\cite{Yang2018}.}
More precisely, the set of admissible control policies is given by
$\Pi:= \{ \pi \mid \pi_t (x_t) = u_t \: \forall t\}$.
To design a controller that is robust against errors in the distribution, 
we employ an (hypothetical) opponent that selects the probability distribution $\mu_t$ in an adversarial way.
The opponent policy $\gamma := (\gamma_0, \gamma_1, \ldots)$ is also assumed to be deterministic and Markov, where $\gamma_t$ maps the current state-input pair $(x_t, u_t)$ to a probability distribution $\mu_t$.
Similarly, the set of admissible opponent's policies is defined by
$\Gamma := \{\gamma \mid \gamma_t (x_t, u_t) = \mu_t \; \forall t\}$.

Suppose for a moment that the controller aims to minimize the standard quadratic cost function $\mathbb{E} [\sum_{t=0}^\infty [x_t^\top Q x_t + u_t^\top R u_t]]$ with $Q = Q^\top \succeq 0$ and $R = R^\top \succ 0$, while the opponent determines $\gamma$ to maximize the same cost. 
If this were the case,
however, this would give too much freedom to the opponent, thereby causing the optimal controller to be overly conservative.  
To systematically adjust conservativeness, we penalize the opponent according to the degree of deviation from the empirical distribution $\nu_t$. 
By doing so, we can also incorporate the prior information provided by the sample data into the controller design. 
Specifically, we modify the cost function as follows:
\begin{equation} \nonumber
\begin{split}
&J_{\bm{x}}(\pi, \gamma) := \\
&\mathbb{E}^{\pi, \gamma} \bigg [
\sum_{t=0}^\infty [x_t^\top Q x_t + u_t^\top R u_t - \lambda W_2 (\mu_t, \nu_t)^2 ] \bigg\vert  x_0 = \bm{x}
\bigg ],
\end{split}
\end{equation}
where $\lambda >0$ is the penalty parameter and the Wasserstein metric $W_2(\mu_t, \nu_t)$ is used to measure the distance between $\mu_t$ and $\nu_t$.  The Wasserstein metric of order $2$ between two distributions $\mu$ and $\nu$ is defined as
\begin{equation*}
\begin{split}
     W_2(\mu,\nu):= \inf_{\eta\in \mathcal{P}(\mathbb{W}^2)}
     \bigg\{ \left ( \int_{\mathbb{W}^2} \|x - y\|^2 \mathrm{d}\eta(x,y) \right ) ^{\frac{1}{2}} \\
     \vert \; \Pi^1\eta=\mu, \Pi^2 \eta=\nu \bigg\},  
\end{split}
\end{equation*}
where $\Pi^i \eta$ is $i$-th marginal distribution of $\eta$ and $\| \cdot \|$ is the standard Euclidean norm.
Thus, by tuning the parameter $\lambda$, we can adjust the conservativeness of our control policy that is obtained by solving the following minimax stochastic control problem:
\begin{equation} \label{opt}
\min_{\pi \in \Pi} \max_{\gamma \in \Gamma} J_{\bm{x}} (\pi, \gamma).
\end{equation}
The inner maximization problem yields a worst-case distribution policy given $\pi$. 
Thus, an optimal solution $\pi_*$ to the outer problem minimizes the worst-case cost and is robust against the deviation of $\mu_t$ from the empirical distribution.
This problem can be viewed as a relaxed version of DRC that explicitly limits the possible range of $\mu_t$ within a \emph{Wasserstein ball} centered at $\nu_t$. More details about the connections to DRC can be found in~\cite{Yang2018}.


\section{Finite-Horizon Case}\label{sec:finite}

To begin with, we consider the minimax control problem in the finite-horizon setting with cost function
\begin{equation*}
\begin{split}
J_{\bm{x}} (\pi, \gamma) := \mathbb{E}^{\pi, \gamma} \bigg[ &\sum_{t=0}^{T-1} [x_t^\top Q x_t + u_t^\top R u_t    -\lambda W_2(\mu_t, \nu_t)^2] \\
&+ x_T^\top Q_f x_T \bigg\vert x_0 = \bm{x} \bigg],
\end{split}
\end{equation*}
where $Q_f = Q_f^\top \succeq 0$.
Later, we establish the connection between the finite-horizon and infinite-horizon cases by letting $T \to \infty$.

We use dynamic programming to solve the finite-horizon problem: 
let $V_t: \mathbb{R}^n \to \mathbb{R}$ be the value function, defined by
$V_t(\bm{x}) :=\inf_{\pi \in \Pi} \sup_{\gamma \in \Gamma}   \mathbb{E}^{\pi, \gamma} [ \sum_{s=t}^{T-1} [x_s^\top Q x_s + u_s^\top R u_s    -\lambda W_2(\mu_s, \nu_s)^2] 
+ x_T^\top Q_f x_T \mid x_t = \bm{x} ]$,
which represents the optimal worst-case expected cost-to-go from stage $t$ given $x_t = \bm{x}$. 
The dynamic programming recursion gives
\begin{equation}\nonumber
\begin{split}
    V_t (\bm{x}) = \bm{x}^\top Q \bm{x} &+ \inf_{\bm{u} \in \mathbb{R}^m} \sup_{\bm{\mu} \in \mathcal{P}(\mathbb{R}^k)} \bigg [ \bm{u}^\top R \bm{u} - \lambda W_2(\bm{\mu}, \nu_t)^2 \\ 
    &+\int_{\mathbb{R}^k} V_{t+1} (A \bm{x} + B \bm{u} + \Xi w) \mathbf{d} \bm{\mu} (w) \bigg ]
\end{split}
\end{equation}
for $t = 0, \ldots, T-1$, and $V_T(\bm{x}) := \bm{x}_T^\top Q_f x_T$.
Note that the inner maximization problem is an infinite-dimensional optimization problem over $\mathcal{P}(\mathbb{R}^k)$.
For a tractable reformulation, we use a modern DRO technique based on Kantorovich duality~\cite{Gao2016}, which yields
\begin{equation}\label{eq:backward}
\begin{split}
    &V_t (\bm{x}) =  \bm{x}^\top  Q \bm{x} + \inf_{\bm{u} \in \mathbb{R}^m} \bigg[ \bm{u}^\top R \bm{u}\: + \\
    &  \frac{1}{N} \sum_{i=1}^N \sup_{ w\in \mathbb{R}^k} \{ V_{t+1}(A \bm{x} 
    + B \bm{u} + \Xi w)-\lambda \lVert \hat{w}^{(i)}_t - w \rVert^2 \} \bigg].
\end{split}
\end{equation}

We assume without loss of generality that
$\mathbb{E}_{\nu_t} [w_t] = 0$ and $\mathbb{E}_{\nu_t} [w_t w_t^\top] = \Sigma_t$ for some covariance matrix $\Sigma_t \in \mathbb{R}^{k \times k}$ for each $t$. 
When the empirical distribution has a non-zero mean, we can normalize it to a zero-mean distribution and obtain an equivalent problem (see~\cite[Appendix E]{Yang2018}).

We now parameterize the value function in quadratic form, $V_t (\bm{x}) = \bm{x}^\top P_t \bm{x} + z_t$, where $P_t$ is a symmetric matrix,  and identify an explicit solution to the minimax optimization problem in \eqref{eq:backward}. 
We then show that the quadratic structure of the value function is preserved through the Bellman recursion, and  the proposed parameterization would thus be exact if  matrices $P_t$ satisfy a Riccati equation.  

\begin{lemma}\label{lem:sol}
Suppose that 
\[
V_{t+1} (\bm{x}) = \bm{x}^\top P_{t+1} \bm{x} + z_{t+1}
\]
 for some $P_{t+1} = P_{t+1}^\top \in \mathbb{R}^{n\times n}$ and $z_{t+1} \in \mathbb{R}$. 
We further assume that the penalty parameter satisfies $\lambda > \bar{\lambda}_{t+1}$, where $\bar{\lambda}_{t+1}$ is the maximum eigenvalue of $\Xi^\top P_{t+1} \Xi$. 
Then, the inner maximization problem 
$\sup_{ w\in \mathbb{R}^k} \{ V_{t+1}(A \bm{x} 
    + B \bm{u} + \Xi w)-\lambda \lVert \hat{w}^{(i)}_t - w \rVert^2 \}$
in \eqref{eq:backward} has a unique maximizer $w_t^\star := (w^{\star,(1)}_{t}, \ldots, w^{\star,(N)}_{t})$, defined by
\begin{equation}\label{w_opt}
w^{\star,(i)}_{t} := (\lambda I - \Xi^\top P_{t+1} \Xi)^{-1} (\Xi^\top P_{t+1} (A\bm{x} + B \bm{u}) + \lambda \hat{w}_t^{(i)}).
\end{equation}
Furthermore, the outer minimization problem in \eqref{eq:backward} has a unique minimizer 
\begin{equation}\label{u_opt}
\begin{split}
&\bm{u}^\star :=\\
& -R^{-1} B^\top \bigg [ I + P_{t+1} B R^{-1} B^\top - \frac{1}{\lambda} P_{t+1} \Xi \Xi^\top \bigg ]^{-1}P_{t+1} A \bm{x}. 
\end{split}
\end{equation}
\end{lemma}

\begin{proof}
See Appendix~\ref{app:lem:sol}.
\end{proof}

Note that $w_t^\star$ is linear in $(\bm{x}, \bm{u})$ and $u_t^\star$ is linear in $\bm{x}$. 
The explicit derivation with this linear structure yields the following Riccati equation:
\begin{equation}\label{ric}
\begin{split}
P_t &= Q + A^\top \bigg [I + P_{t+1} BR^{-1} B^\top - \frac{1}{\lambda} P_{t+1} \Xi \Xi^\top \bigg ]^{-1} P_{t+1} A\\
z_t &= z_{t+1} + \mathrm{tr} \bigg [
\bigg (
I - \frac{1}{\lambda} \Xi^\top P_{t+1} \Xi
\bigg )^{-1} 
\Xi^\top P_{t+1} \Xi \Sigma_t
\bigg ]
\end{split}
\end{equation}
with $P_T := Q_f$ and $z_T := 0$.
Note that $P_t$ are symmetric since $P_T$ is symmetric.
For the well-definedness of the recursion, we make the following assumption:
\begin{assumption}\label{ass:pen}
The penalty parameter satisfies $\lambda > \bar{\lambda}_{t}$ for all $t$, where $\bar{\lambda}_{t}$ is the maximum eigenvalue of $\Xi^\top P_{t} \Xi$.
\end{assumption}

\begin{theorem}[Optimal policy]\label{thm:fin} 
Suppose that Assumption~\ref{ass:pen} holds.
Then, the matrices $P_t$ are well-defined and 
the value function can be expressed as
\[
V_t(\bm{x}) = \bm{x}^\top P_t \bm{x} + z_t.
\]
Furthermore, the problem~\eqref{opt} in the finite-horizon case has a unique optimal policy, defined by
\begin{equation} \label{opt_policy}
\begin{split} 
&\pi^\star_t (\bm{x}) := K_t \bm{x}, \quad t = 0, \ldots, T-1,
\end{split}
\end{equation}
where $K_t := -R^{-1} B^\top [ I + P_{t+1} B R^{-1} B^\top - P_{t+1} \Xi \Xi^\top /\lambda  ]^{-1}$ $P_{t+1} A$.
\end{theorem}

\begin{proof}
See Appendix~\ref{app:thm:fin}.
\end{proof}

As in the standard LQG, the optimal policy is linear in system state and the  gain matrix $K_t$ can be obtained by solving a Riccati equation. 
Note that the Riccati equation in the standard LQG is given by (e.g.,~\cite{Astrom2012})
\begin{equation}\label{ric_lqg}
\begin{split}
P_t &= Q + A^\top   (I + P_{t+1} B R^{-1}B^\top )^{-1} P_{t+1} A\\
z_t &= z_{t+1} + \mathrm{tr}[ \Xi^\top P_{t+1} \Xi \Sigma_t ],
\end{split}
\end{equation}
and it can be obtained by letting $\lambda \to \infty$ in \eqref{ric}.
Increasing $\lambda$ encourages the opponent not to deviate much from the empirical distribution $\nu_t$. Thus, in the limit,  our minimax method is equivalent to the standard LQG. 
This shows that our proposed framework is a generalization of LQG.

Another immediate consequence of Lemma~\ref{lem:sol} and Theorem~\ref{thm:fin} is that one of the worst-case distributions can be explicitly obtained with a finite support, as follows.

\begin{cor}[Worst-case distribution]\label{cor:dist}
Suppose that Assumption~\ref{ass:pen} holds. 
Let 
\begin{equation} \nonumber
\begin{split}
&w_t^{\star, (i)} (\bm{x}) \\
&:= (\lambda I - \Xi^\top P_{t+1} \Xi)^{-1} (\Xi^\top P_{t+1} (A + B K_t) \bm{x} + \lambda \hat{w}_t^{(i)}).
\end{split}
\end{equation}
Then, the policy $\gamma^\star$ defined by
\[
\gamma^\star_t (\bm{x}) := \frac{1}{N} \sum_{i=1}^N \delta_{w_t^{\star, (i)} (\bm{x})} 
\]
generates the worst-case distribution, i.e., $(\pi^\star, \gamma^\star)$ is an optimal minimax solution to \eqref{opt} in the finite-horizon case.
\end{cor}

\section{Infinite-Horizon Case}\label{sec:infinite}

We now consider the original infinite-horizon case.
Based on the results in the finite-horizon case, our goal is to derive an algebraic Riccati equation (ARE) and characterize the condition under which the recursion~\eqref{ric} converges to 
 a unique symmetric PSD solution of the ARE. 
Throughout this subsection, we assume the following for the stationarity of the problem.
\begin{assumption}\label{ass:st}
The random disturbance process $\{w_t\}_{t=0}^\infty$ is i.i.d.,
and its empirical distribution is constructed as
$\nu \equiv \nu_t := \frac{1}{N} \sum_{i=1}^N \delta_{\hat{w}^{(i)}}$ from the dataset $\{ \hat{w}^{(1)}, \ldots, \hat{w}^{(N)} \}$.
\end{assumption}

Define an $n \times n$ matrix $W$ as
\[
W := BR^{-1} B^\top - \frac{1}{\lambda} \Xi \Xi^\top.
\] 
We make the following assumption:

\begin{assumption} \label{ass:W}
$W \succeq 0$, and $(A, \sqrt{W})$ is stabilizable.
\end{assumption}

\begin{theorem}\label{thm:are}
Suppose that Assumptions~\ref{ass:pen}--\ref{ass:W} hold. 
Then, a bounded limiting solution $P_{ss}:= \lim_{T \to \infty} P_t$ to the Riccati equation~\eqref{ric} exists for any symmetric PSD $P_T$.
Furthermore, $P_{ss}$ is a symmetric PSD solution to the following (discrete) ARE: 
\begin{equation}\label{are}
P = Q + A^\top  \bigg [ I + PBR^{-1} B^\top - \frac{1}{\lambda} P \Xi \Xi^\top \bigg ]^{-1} PA.
\end{equation}
\end{theorem}
\begin{proof}
In the standard LQG, it is well known that if $(A, B)$ is stabilizable, the Riccati equation~\eqref{ric_lqg} has a bounded limiting solution, which coincides with a symmetric PSD solution to an associated ARE~\cite[Theorem 2.4-1]{Lewis2012}.
By observing that our ARE~\eqref{are} is obtained by replacing $(A, B)$ with $(A, \sqrt{W})$ in the ARE for the standard LQG, the result follows.
\end{proof}

\subsection{Connecting the  Infinite-Horizon and  the Finite-Horizon Cases}

We now consider the mean-state dynamics, which is equivalent to the Hamiltonian system of deterministic LQR~\cite{Lewis2012}.

\begin{proposition}\label{prop:mean}
Let $\bar{x}_0 := x_0^\star$ and $\bar{x}_t := \mathbb{E} [x_t^\star]$ for $t = 1, 2, \ldots$, where $x_t^\star$ denotes the closed-loop system state under the optimal policy in Theorem~\ref{thm:fin}. Then, we have
\[
F\begin{bmatrix} \bar{x}_t\\ P_t \bar{x}_t \end{bmatrix}
    = G\begin{bmatrix} \bar{x}_{t+1}\\ P_{t+1} \bar{x}_{t+1} \end{bmatrix}, \quad t = 0, 1, \ldots, 
\]
where $F := \begin{bmatrix} A & 0 \\ -Q & I\end{bmatrix}$ and $G :=\begin{bmatrix} I & W \\ 0 & A^\top\end{bmatrix}$.
\end{proposition}
\begin{proof}
See Appendix~\ref{app:prop:mean}.
\end{proof}

To solve the ARE~\eqref{are}, we use the method proposed in~\cite{Pappas1980}, considering the generalized eigenvalue problem of $F$ and $G$, 
\begin{equation}\label{gen}
F v = \gamma G v.
\end{equation}

\begin{lemma}\label{lem:sol2}
Any solution of the ARE~\eqref{are} can be expressed as
\[
P = \hat{V}_2 \hat{V}_1^{-1},
\]
where each column of $\begin{bmatrix} \hat{V}_1 \\ \hat{V}_2 \end{bmatrix} \in \mathbb{R}^{2n \times n}$ solves the generalized eigenvalue problem~\eqref{gen} of $F$ and $G$.
\end{lemma}
\begin{proof}
See Appendix~\ref{app:lem:sol2}.
\end{proof}

Lemma~\ref{lem:sol2} shows that all solutions of the ARE~\eqref{are} can be obtained from the generalized eigenvalue problem of $F$ and $G$. 
Unfortunately, most of them are unstabilizing solutions.
However, we are only interested in the symmetric PSD solution $P_{ss}$ to which the Riccati recursion~\eqref{ric} converges. 
To identify the steady-state solution, we need the following assumption and lemma. 

\begin{assumption}\label{ass:ob}
$(A, \sqrt{Q})$ is observable. 
\end{assumption}

\begin{lemma}\label{lem:stable}
Suppose that Assumptions~\ref{ass:W} and \ref{ass:ob} hold.
Then, $P = \hat{V}_2 \hat{V}_1^{-1}$ is a symmetric PSD solution to the ARE~\eqref{are} if and only if each column of $\begin{bmatrix} \hat{V}_1 \\ \hat{V}_2 \end{bmatrix} \in \mathbb{R}^{2n \times n}$ solves the generalized eigenvalue problem~\eqref{gen} of $F$ and $G$ with a stable generalized eigenvalue.\footnote{A generalized eigenvalue is stable if its absolute value is less than $1$.} 
\end{lemma}

\begin{proof}
See Appendix~\ref{app:lem:stable}.
\end{proof}

Lemma~\ref{lem:stable} motivates us to investigate the condition on $F$ and $G$ under which \eqref{gen} has $n$ stable generalized eigenvalues. 
Note that  the following symplectic property holds
\[
F \Omega F^\top = G \Omega G^\top  = \begin{bmatrix} 0 & A \\ -A^\top & 0\end{bmatrix},
\]
where $\Omega = \begin{bmatrix} 0 & I_n \\ -I_n & 0\end{bmatrix}$.
Thus, if $\gamma$ is a generalized eigenvalue, so is ${1}/{\gamma}$ with the same multiplicity.
This implies that if no generalized eigenvalue lies on the unit circle, then exactly $n$ generalized eigenvalues are stable, and there exists a unique symmetric PSD solution to the ARE by Lemma~\ref{lem:stable}.

\begin{lemma}\label{lem:circ}
Under Assumptions~\ref{ass:W} and \ref{ass:ob}, $(F, G)$ does not have any generalized eigenvalue on the unit circle. 
\end{lemma}
\begin{proof}
The existence of generalized eigenvalues on the unit circle contradicts Assumptions \ref{ass:W} and \ref{ass:ob}. See \cite[Theorem 3]{Pappas1980} for details.
\end{proof}

By Lemma~\ref{lem:circ}, there exist $V_1, V_2 \in \mathbb{R}^{n\times n}$ and $M\in \mathbb{R}^{n\times n}$ such that 
\begin{equation}\label{gen2}
FV=GVM
\end{equation}
 with $V=\begin{bmatrix}V_{1}\\V_{2}\end{bmatrix}$, where 
the  columns of $V$ solve \eqref{gen} with $n$ stable generalized eigenvalues, and 
$M$ is the corresponding Jordan normal form. 
We obtain the following lemma that yields to construct a solution of the ARE~\eqref{are} from $V_1$ and $V_2$. 

\begin{lemma}
Under Assumptions~\ref{ass:W} and \ref{ass:ob},  $V_1$ is nonsingular. 
\end{lemma}
\begin{proof}
This can be shown directly using the proof of \cite[Theorem 6]{Pappas1980}.
\end{proof}

Using the previous lemmas, we finally obtain the following conclusion that connects the Riccati equation~\eqref{ric} in the finite-horizon case and the ARE~\eqref{are} in the infinite-horizon. 

\begin{theorem}
Suppose that Assumptions~\ref{ass:pen}--\ref{ass:ob} hold. 
Then, the recursion~\eqref{ric} converges to the unique symmetric PSD solution $P_{ss} := V_2 V_1^{-1}$ of the ARE~\eqref{are}.
\end{theorem}

This result can further be simplified 
when the system matrix $A$ is nonsingular.
In this particular case, we have
\begin{equation}\nonumber
    \begin{bmatrix}\bar{x}_{t+1}\\ P_{t+1} \bar{x}_{t+1}\end{bmatrix} = H'\begin{bmatrix}\bar{x}_{t} \\ P_{t} \bar{x}_{t}
     \end{bmatrix},
\end{equation}
where
\begin{equation*}
    H' := G^{-1}F =\begin{bmatrix} A + W A^{-\top} Q & - W A^{-\top} \\ - A^{-\top} Q & A^{-\top} \end{bmatrix}.
\end{equation*}
Note that the matrix $H'$ corresponds to the inverse of the Hamiltonian matrix in the standard LQR.
We construct $U_1, U_2 \in \mathbb{R}^{n \times n}$ so that each column of $\begin{bmatrix} U_{1} \\ U_{2} \end{bmatrix} \in \mathbb{R}^{2n \times n}$ is an eigenvector of $H'$ corresponding to a stable eigenvalue. 
We then obtain the following result:

\begin{cor}
Suppose that Assumptions~\ref{ass:pen}--\ref{ass:ob} hold and that $A$ is nonsingular. 
Then, the recursion~\eqref{ric} converges to the unique symmetric PSD solution $P_{ss} := U_2 U_1^{-1}$ of the ARE~\eqref{are}.
\end{cor}

\subsection{Closed-Loop Stability}

The optimal control policy in the infinite-horizon case
can be obtained using the symmetric PSD solution to the ARE~\eqref{are} as in the finite-horizon case.

\begin{cor}[Optimal minimax solution]
Suppose that Assumptions~\ref{ass:pen}--\ref{ass:ob} hold.
Then, the problem~\eqref{opt} in the infinite-horizon case has a unique optimal policy, defined by
\[
\pi^\star (\bm{x}) := K_{ss} \bm{x},
\]
where $K_{ss} := -R^{-1} B^\top [ I + P_{ss} B R^{-1} B^\top - P_{ss} \Xi \Xi^\top /\lambda  ]^{-1}$ $P_{ss} A$. 
Furthermore, the policy $\gamma^\star$ defined by
\begin{equation}\label{eq:steady}
\gamma^\star (\bm{x}) := \frac{1}{N} \sum_{i=1}^N \delta_{w^{\star, (i)} (\bm{x})} 
\end{equation}
with
$w^{\star, (i)} (\bm{x}) 
:= (\lambda I - \Xi^\top P_{ss} \Xi)^{-1} (\Xi^\top P_{ss} (A + B K_{ss}) \bm{x} + \lambda \hat{w}^{(i)})$
generates the worst-case distribution, i.e., $(\pi^\star, \gamma^\star)$ is an optimal minimax solution to \eqref{opt} in the infinite-horizon case.
\end{cor}

We show that the derived control policy achieves closed-loop stability in the following sense:

\begin{theorem}
Suppose that Assumptions~\ref{ass:pen}--\ref{ass:ob} hold.
Then, the optimal policy $\pi^\star$ stabilizes the expected state of the stochastic system under the worst-case distribution generated by $\gamma^\star$. 
\end{theorem}

\begin{proof}
By \eqref{gen2},   we have $A = V_1 M V_1^{-1} + W V_2 M V_1^{-1}$.
Furthermore, it is shown in Appendix~\ref{app:prop:mean} that
\begin{equation*}
\begin {split}
    \bar{x}_{t+1} &= A \bar{x}_t - B R^{-1} B^\top \bar{g}_t  + \frac{1}{\lambda} \Xi \Xi^\top \bar{g}_t\\
    &= \{ A - W (I + P_{ss} W)^{-1} P_{ss} A  \} \bar{x}_{t}.
\end{split}
\end{equation*}
Then, we have
\begin{equation*}
\begin{split}
    &A - W (I + P_{ss} W)^{-1} P_{ss} A = (I + W P_{ss})^{-1} A\\
    &= (I + W V_2 V_1^{-1})^{-1} (V_1 M V_1^{-1} + W V_2 M V_1^{-1})\\
    &= V_1 M V_1^{-1},
\end{split}
\end{equation*}
which implies that $A - W (I + P_{ss} W)^{-1} P_{ss} A$ and $M$ have the same spectrum.
Therefore, the mean-state system is stable since $M$ has $n$ stable eigenvalues.
\end{proof}

The stabilizing optimal policy can be obtained by solving the generalized eigenvalue problem of $F$ and $G$. 
However, the numerical solution of the generalized eigenvalue problem involves inefficient computations.
Instead, the Schur decomposition or the QZ algorithm can be used to directly find the symmetric PSD solution of the ARE~\eqref{are}~\cite{Pappas1980}.

\section{Relations to $H_\infty$-Optimal Control}\label{sec:robust}

In this section, we discuss relations between our minimax control method and the $H_\infty$-method.
Specifically, we are concerned with the dynamic game formulation of the $H_\infty$-optimal control, which is investigated in~\cite{Basar2008}.
We consider the problem of minimizing the $H_\infty$-norm of the cost function with respect to the disturbance.

To begin with, we examine the finite-horizon case with the initial state being fixed as zero, i.e., $x_0 = 0$. 
For $H_\infty$-control, we adopt a similar but modified dynamic game formulation, where the opponent's policy $\tilde{\gamma}_t$ now maps the current state $x_t$ to disturbance vector $w_t$ rather than its distribution. Note that the disturbance vector is no longer random in the $H_\infty$-setting. 
The set of admissible opponent's policies is accordingly modified and is denoted by $\tilde{\Gamma}$.
Consider the following quadratic cost function:
\[
\tilde{J} (\pi, \tilde{\gamma}) := \sum_{t=0}^{T-1} ( x_t^\top Q x_t + u_t^\top R u_t ) + x_T^\top Q_f x_T.
\]
We aim to find the minimal value of $\lambda > 0$ such that 
\begin{equation*}
   \sup_{\tilde{\gamma} \in \tilde{\Gamma} : \lVert w \rVert \leq 1} \tilde{J} (\pi, \tilde{\gamma}) = \sup_{\tilde{\gamma} \in \tilde{\Gamma} } \frac{\tilde{J}(\pi, \tilde{\gamma})}{\| w \|^2 } \leq \lambda,
\end{equation*}
given a control policy $\pi$, 
where $w:= (w_0, w_1, \ldots, w_{T-1})$ and $\| w \|^2 := \sum_{t=0}^{T-1} \lVert w_t \rVert^2$. 
The equality holds because $\tilde{J} (\pi, \tilde{\gamma})$ is homogeneous with respect to $\| w\|^2$ when $x_0 = 0$. 
Note that ${\tilde{J}(\pi, \tilde{\gamma})}/{\sum_{t=0}^{T-1} \lVert w_t \rVert^2} \leq \lambda$ for all $\tilde{\gamma} \in \tilde{\Gamma}$ if and only if $\tilde{J}(\pi, \tilde{\gamma}) - \lambda \sum_{t=0}^{T-1} \lVert w_t \rVert^2 \leq 0$ for all $\tilde{\gamma} \in \tilde{\Gamma}$.
Thus, the inequality part can be rewritten as
\[
 \sup_{\tilde{\gamma} \in \tilde{\Gamma} } \bigg [
 \tilde{J}(\pi, \tilde{\gamma}) - \lambda \sum_{t=0}^{T-1} \| w_t \|^2
 \bigg ]\leq 0.
\]
This motivates us to consider the following augmented cost function with an additional disturbance-norm term on each stage:
\[
J_\infty^{\lambda} (\pi, \tilde{\gamma}) := \sum_{t=0}^{T-1} 
( x_t^\top Q x_t + u_t^\top R u_t - \lambda \| w_t \|^2 ) + x_T^\top Q_f x_T,
\]
as well as the following minimax control problem:
\[
J_\infty^{\lambda, \star} := \inf_{\pi \in \Pi} \sup_{\tilde{\gamma} \in \tilde{\Gamma}} J_\infty^{\lambda} (\pi, \tilde{\gamma}).
\]
Let $\Lambda:= \{ \lambda \mid J_\infty^{\lambda, \star} \leq 0\}$.
Then, we can find the desired $\lambda^\star$ as 
$\lambda^\star := \inf \{ \lambda \mid \lambda \in \Lambda \}$.
More details about the dynamic game formulation of $H_\infty$-control can be found in~\cite[Section 1.4]{Basar2008}.
Let $\tilde{V}_t:\mathbb{R}^n \to \mathbb{R}$ denote the value function of this problem. The dynamic programming principle yields
\begin{equation}\nonumber
\begin{split}
&\tilde{V}_t (\bm{x}) = \bm{x}^\top Q \bm{x} +\\
& \inf_{\bm{u} \in \mathbb{R}^m} \bigg [
\bm{u}^\top R \bm{u} + \sup_{\bm{w} \in \mathbb{R}^k} \{
\tilde{V}_{t+1} (A\bm{x} + B\bm{u} + \Xi \bm{w}) - \lambda \| \bm{w} \|^2
\}
\bigg ]
\end{split}
\end{equation}
with $\tilde{V}_T (\bm{x}) := \bm{x}^\top Q_f \bm{x}$. 
Under Assumption~\ref{ass:pen}, 
we can show that the optimal value function, optimal control policy and the Riccati equation are given by the same results as ours in Theorem~\ref{thm:fin}, except that there is no $z_t$ term in the $H_\infty$-control. 
The worst-case disturbance policy is then given by
\[
\tilde{\gamma}_t^\star (\bm{x}) := (\lambda I - \Xi^\top P_{t+1} \Xi)^{-1} \Xi^\top P_{t+1} (A+ BK_t) \bm{x}.
\]
Recall that $x_0 = 0$, which implies 
\[
J_\infty^{\lambda, \star} = \tilde{V}_0 (0) = 0^\top P_0 0 = 0.
\]
In this case, any $\lambda$ satisfying Assumption~\ref{ass:pen} should be contained in $\Lambda$. 
However, if $\lambda$ does not satisfy Assumption~\ref{ass:pen}, then the cost value will be $+\infty$ and, therefore, $\lambda$ cannot belong to $\Lambda$.
Thus, we conclude that $\lambda^\star$ is the infimum of $\lambda$ that satisfies Assumption~\ref{ass:pen}.

The worst-case disturbance in the $H_\infty$-method is related with the support element of the worst-case distribution in Corollary~\ref{cor:dist} in our method through 
\[
w_t^{\star, (i)} (\bm{x}) = \tilde{\gamma}_t^\star (\bm{x})  + \bigg [
I - \frac{1}{\lambda} \Xi^\top P_{t+1} \Xi
\bigg ]^{-1} \hat{w}^{(i)}_t.
\]
This indicates that each support element of the worst-case distribution in Corollary~\ref{cor:dist} can be considered to be shifted from $\tilde{\gamma}_t^\star (\bm{x})$ by the scaled sample data $\hat{w}^{(i)}_t$.  
As the sample mean is assumed to be zero, $\tilde{\gamma}_t^\star (\bm{x})$ is the mean value of the worst-case distribution.
Thus, our minimax control method with Wasserstein distance can be understood as a distributional generalization of the $H_\infty$-method.

In the infinite-horizon case, the corresponding $H_\infty$-control can be obtained using a limiting solution of the Riccati equation. 
We obtain the same ARE as \eqref{are} for our minimax control method~\cite[Section 3.4]{Basar2008}. 
Under Assumptions~\ref{ass:pen}--\ref{ass:ob}, the ARE has a symmetric PSD solution $P_{ss}$ from which we can obtain the same optimal control gain $K_{ss}$ and optimal policy.
Regarding the worst-case disturbance, we have 
\[
\tilde{\gamma}^\star (\bm{x}) := (\lambda I - \Xi^\top P_{ss} \Xi)^{-1} \Xi^\top P_{ss} (A+ BK_{ss}) \bm{x}.
\]
Thus, 
 the worst-case disturbance in the $H_\infty$ method is related to our steady-state worst-case distribution as follows: 
\[
w^{\star, (i)} (\bm{x}) = \tilde{\gamma}^\star (\bm{x})  + \bigg [
I - \frac{1}{\lambda} \Xi^\top P_{ss} \Xi
\bigg ]^{-1} \hat{w}^{(i)}.
\]
Therefore, we obtain the same relationship as the one in the finite-horizon case.

\section{Numerical Experiments}\label{sec:exp}

In this section, we demonstrate the effectiveness of our minimax control method on a power system frequency regulation problem.
Stability is a major concern in power transmission systems, as the penetration of variable renewable energy sources 
and the potential of data integrity attacks
 increase.
We apply the minimax control method on the IEEE 39 bus system, which models the New England power grid and has been frequently used to evaluate frequency control methods (e.g. \cite{Dorfler2014, Dizche2019}).
This model consists of 39 buses, 46 lines, and 10 generators.
We use a classical generator model without an excitation system, such as a power system stabilizer  and an automatic voltage regulator, for simplicity.

Let $\delta_i$ and  $\omega_i$ denote the rotor angle and the frequency of the $i$th generator. They satisfy $\dot{\delta_i} = \omega_i - \omega_s$, where $\omega_s$ is a constant synchronous speed. 
The electromechanical swing equation for  the $i$th generator is given by a damped oscillator as follows:
\begin{equation*}
    \frac{2H_i}{\omega_s} \dot{\omega_i} =  T_{i} - d_i \omega_i - \sum_{j \neq i} \vert Y_{ij} \vert E_i E_j \sin(\delta_i - \delta_j),
\end{equation*}
where
$H_i$, $T_i$, $d_i$, and $E_i$ denote the inertia, the power injection, the damping coefficient, and the voltage of the $i$th generators, respectively, and
$Y$ denotes the admittance matrix of the power network~\cite{Sauer1998}.
By linearizing the equations at an operating point $(\delta^*, \omega^*)$, we obtain
\begin{equation*}
    M \Delta \ddot{\delta} + D \Delta\dot{\delta} + L \Delta\delta = \Delta P,
\end{equation*}
where
$M := \mathrm{diag}(2H_1/\omega_s, \ldots, 2H_{10}/\omega_s)$, $D := \mathrm{diag}(d_1, \ldots, d_{10})$, and  the matrix $L$ is defined by $L_{ij} := - \vert Y_{ij} \vert E_i E_j \cos(\delta^*_i-\delta^*_j)$ for $i \neq j$ and $L_{ii} := -\sum_{j \neq i} L_{ij}$.
The second-order ordinary differential equation can be expressed in the following state-space form:
\begin{equation*}
    \begin{bmatrix}\Delta \dot{\delta} \\\Delta \dot{\omega} \end{bmatrix} = \underbrace{\begin{bmatrix} 0 & I \\ -M^{-1} L & -M^{-1} D \end{bmatrix}}_{=: A} \begin{bmatrix} \Delta \delta \\ \Delta \omega \end{bmatrix}+ \underbrace{\begin{bmatrix} 0 \\ M^{-1} \end{bmatrix}}_{=:B}\Delta P,
\end{equation*}
with system state $x(t):=(\Delta\delta^\top(t), \Delta\omega^\top(t))^\top \in \mathbb{R}^{20}$ and control input $u(t):=\Delta P(t) \in \mathbb{R}^{10}$.

\begin{figure}[t!]
\centering
\includegraphics[scale=0.32]{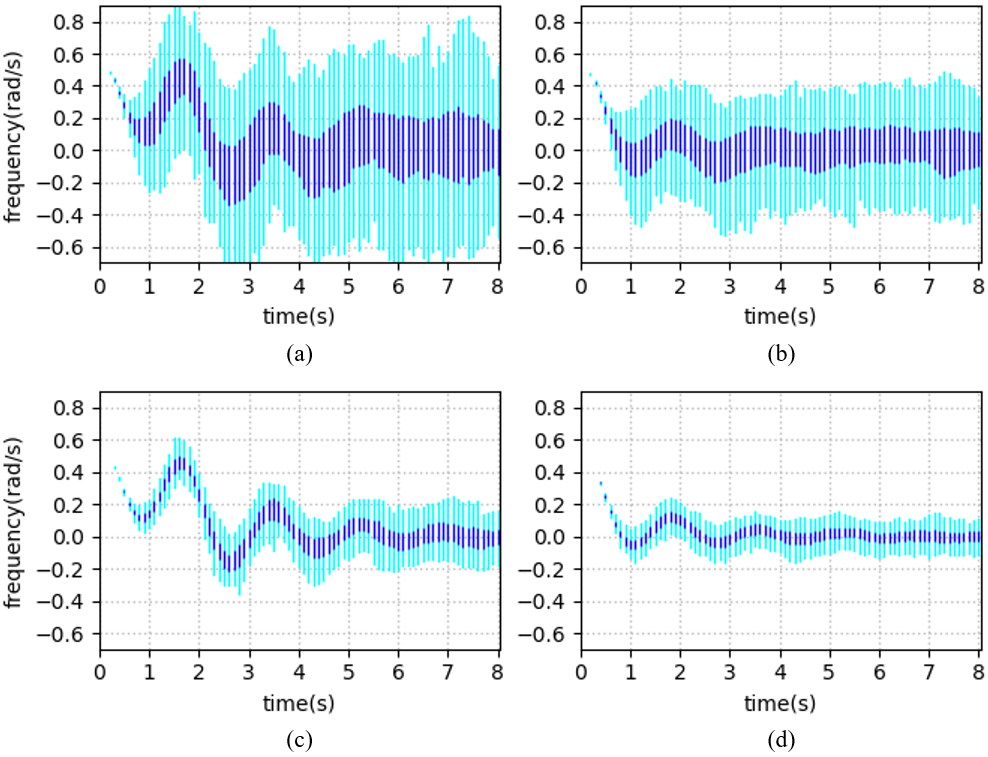}
\caption{Box plots of $\Delta \omega_{10}$, controlled by
(a) the standard LQG method under the worst-case distribution generated with $\lambda=1.29$,
(b) our minimax method under the worst-case distribution generated with $\lambda=1.29$,
(c) the standard LQG method under the worst-case distribution generated with $\lambda=1.30$
and
(d) our minimax method under the worst-case distribution generated with $\lambda=1.30$.}
\label{fig:result}
\vspace{-0.1in}
\end{figure}

We now consider the situation in which a disturbance $w(t)$ is added to the input $u(t)$ to model uncertainty in power injection or net demand.
Then, $\Xi = B$.
For the quadratic cost function, we set $x^\top Q x= \frac{1}{2} \Delta \delta^\top (I_{10}-\frac{1}{10}\mathbf{1}_{10} \mathbf{1}_{10}^\top) \Delta \delta + \frac{1}{2} \Delta \omega^\top \Delta \omega$ and $R=I_{10}$,
where $I_{10}$ denotes the 10 by 10  identity matrix and $\mathbf{1}_{10}$ denotes the 10 dimensional vector of all ones. 
The system is discretized by a zero-order hold method with sampling time $0.1$ seconds.
Suppose that the initial value of rotor speed $\Delta \omega_{10}$ is perturbed by $0.5$, $10$ samples of disturbances are generated according to the normal distribution $\mathcal{N}(0, 0.1^2 I)$, and the worst-case distribution \eqref{eq:steady} is applied to the system.

Fig.~\ref{fig:result} shows the box plot of $100$ test cases for the frequency $\Delta \omega_{10}$, controlled by the standard LQG and the proposed minimax control methods. 
In this setting, the  penalty parameter $\lambda$ should be larger than $1.283$ to satisfy Assumptions \ref{ass:pen} and \ref{ass:W}.
The results demonstrate that our minimax method significantly reduces the fluctuation of the frequency compared to the standard LQG method.
The results also show that the value of $\lambda$ plays an important role in the performance of our method.
As $\lambda$ gets closer to its minimum possible value, the derived policy is robust against a wider range of distributions, and the worst-case distribution is considered to be a more extreme case.
As $\lambda$ increases, the worst-case distribution converges to the empirical distribution, and thus the robustness of our policy diminishes.

\begin{figure}[t!]
\centering
\includegraphics[scale=0.5]{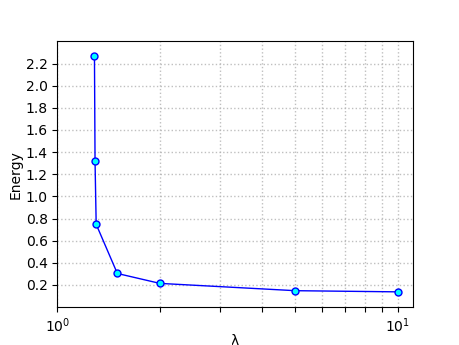}
\caption{Average control energy depending on the value of $\lambda$.}
\label{fig:result2}
\vspace{-0.1in}
\end{figure}

Fig.~\ref{fig:result2} shows the average control  energy required for our method depending on the value of $\lambda$.
The control energy is measured for the first 5 seconds, i.e., $\sum_{t=0}^{49}\lVert u_t \rVert^2 / 50$, and  is averaged over 100 test cases. 
As shown in Fig.~\ref{fig:result2}, the required energy declines as $\lambda$ increases, and eventually converges to the energy required for the standard LQG method.
This implies that a tradeoff between robustness and control energy exists in our minimax method. 
Therefore, the value of $\lambda$ should  be properly selected based on the reliability of available data to balance robustness and control energy.

\section{Conclusions and Future Work}

We have presented a minimax linear-quadratic control method with Wasserstein distance to address the issue of ambiguity inherent in practical stochastic systems. 
Our method has several salient features including $(i)$ a closed-form expression of the unique optimal policy, $(ii)$ the convergence of a Riccati equation to the unique symmetric PSD solution to the corresponding ARE under stabilizability and observability conditions, and $(iii)$ closed-loop stability. 
The relation to the $H_\infty$-method indicates that our method can provide an exciting avenue for future research that connects stochastic and robust control. Moreover, it remains as future work to address partial observability and extensions to continuous-time settings. 

\appendices

\section{Proof of Lemma~\ref{lem:sol}}\label{app:lem:sol}

Since the expression of optimal $w_t^\star$ can be derived using the proof of \cite[Theorem 4]{Yang2018}, we have omitted the detailed proof. Note that the following equality holds:
\begin{equation}\label{eq1}
w_t^{\star, (i)} = \frac{1}{2\lambda} \Xi^\top V_{t+1}' (A \bm{x} + B\bm{u} + \Xi w_t^{\star, (i)}  )+ \hat{w}_t^{(i)}.
\end{equation}

To solve the outer minimization problem in \eqref{eq:backward},
we first take the derivative of the outer objective function with respect to $\bm{u}$ to obtain that, by \eqref{eq1},
\begin{equation} \nonumber
\begin{split}
&2R \bm{u} + \frac{1}{N} \sum_{i=1}^N \bigg [
\bigg (B + \Xi \frac{\partial w_t^{\star, (i)}}{\partial \bm{u}}  \bigg )^\top \times\\
& V_{t+1}' (A\bm{x} + B\bm{u} + \Xi w_t^{\star, (i)} )
+ 2\lambda \frac{\partial w_t^{\star, (i)}}{\partial \bm{u}}^\top (\hat{w}_t^{(i)} - w_t^{\star, (i)})
\bigg ]\\
&= 2R \bm{u} + 2 B^\top g_t (\bm{u}),
\end{split}
\end{equation}
where 
\begin{equation} \label{g_eq}
\begin{split}
g_t (\bm{u}) &:= \frac{1}{2N} \sum_{i=1}^N V_{t+1}' (A\bm{x} + B\bm{u} + \Xi w_t^{\star, (i)}(\bm{u}) )\\
&= P_{t+1} \bigg ( 
A \bm{x} + B\bm{u} + \frac{1}{N} \sum_{i=1}^N \Xi w_t^{\star, (i)}(\bm{u})
\bigg ).
\end{split}
\end{equation}
The second-order derivative of the outer objective function with respect to $\bm{u}$ is then given by
\[
2 \big [ 
R + B^\top P_{t+1} B + B^\top P_{t+1} \Xi (\lambda I - \Xi^\top P_{t+1} \Xi)^{-1} \Xi^\top P_{t+1} B
\big ],
\]
which is positive definite by the assumption on the penalty parameter. 
Thus, the outer objective function is strictly convex, and it has a unique minimizer, $\bm{u}^\star$. 
Equating the derivative to zero yields
\begin{equation}\label{u_eq}
\bm{u}^\star = -R^{-1} B^\top g_t^\star,
\end{equation}
where $g_t^\star := g_t (\bm{u}^\star)$.
By the definition of $g_t$ and~\eqref{eq1}, 
\[
g_t^\star = P_{t+1} \bigg ( 
A \bm{x} - BR^{-1} B^\top g_t^\star  + \frac{1}{\lambda} \Xi \Xi^\top g_t^\star
\bigg ),
\]
which yields the following expression of $g_t^\star$:
\begin{equation}\label{g_eq2}
g_t^\star = \bigg  (
I + P_{t+1} B R^{-1} B^\top  -  \frac{1}{\lambda} P_{t+1} \Xi \Xi^\top
\bigg )^{-1} P_{t+1} A \bm{x}.
\end{equation}
Note that $I + P_{t+1} B R^{-1} B^\top  -  \frac{1}{\lambda} P_{t+1} \Xi \Xi^\top$ must be invertible by the uniqueness of $\bm{u}^\star$.
Therefore, the result follows. \qed

\section{Proof of Theorem~\ref{thm:fin}}\label{app:thm:fin}

We use mathematical induction to show that $V_t(\bm{x}) = \bm{x}^\top P_t \bm{x} + z_t$. 
For $t = T$, the statement is true by the definition of $P_T$ and $z_T$. 
Suppose that the induction statement holds for $t+1$, i.e., $V_{t+1}(\bm{x}) = \bm{x}^\top P_{t+1} \bm{x} + z_{t+1}$.
Recall that $g_t^\star :=g_t (\bm{u}^\star)$, where $g_t$ is given as~\eqref{g_eq}. 
 We first differentiate \eqref{eq:backward} with respect to $\bm{x}$ to obtain that, by \eqref{eq1} and \eqref{u_eq},
 \begin{equation}\nonumber
 \begin{split}
& V_t' (\bm{x}) = 2 Q\bm{x} + 2 \frac{\partial \bm{u}^\star}{\partial \bm{x}}^\top R  \bm{u}^\star +
 \\
 & \frac{1}{N} \sum_{i=1}^N \bigg [
 \bigg (
 A + B \frac{\partial \bm{u}^\star}{\partial \bm{x}} + \Xi \frac{\partial \bm{w}^{\star, (i)}}{\partial \bm{x}}
 \bigg )^\top \times \\
 &V_{t+1}' ( A \bm{x} + B \bm{u}^\star + \Xi \bm{w}^{\star, (i)} )
 +2 \lambda \frac{\partial \bm{w}^{\star, (i)}}{\partial \bm{x}}^\top (\hat{w}_t^{(i)} - \bm{w}^{\star, (i)} )
 \bigg ]\\
 &= 2 Q\bm{x} + A^\top \frac{1}{N} \sum_{i=1}^N V_{t+1}' ( A \bm{x} + B \bm{u}^\star + \Xi \bm{w}^{\star, (i)} )  \\
& + \frac{\partial \bm{u}^\star}{\partial \bm{x}}^\top  \bigg [2R \bm{u}^\star  + B^\top \frac{1}{N} \sum_{i=1}^N V_{t+1}' ( A \bm{x} + B \bm{u}^\star + \Xi \bm{w}^{\star, (i)} ) \bigg ]\\
&=2Q\bm{x} + 2A^\top g_t^\star,
 \end{split}
 \end{equation}
 where $\bm{u}^\star$ is given as~\eqref{u_opt} and $\bm{w}^{\star, (i)}$ is given as~\eqref{w_opt} with $\bm{u} := \bm{u}^\star$.
  Replacing $g_t^\star$ with \eqref{g_eq2}, we have
  \begin{equation}\label{eq_Q}
   Q\bm{x} + A^\top g_t^\star = P_t \bm{x}
  \end{equation}
  by the recursion for $P_t$ in the Riccati equation~\eqref{ric}.
  Thus, 
  \[
  \frac{1}{2} V_t'(\bm{x}) = P_t \bm{x}.
  \]
 Using the recursion for $z_t$ in \eqref{ric}, we can conclude that
 \[
  V_t (\bm{x}) = \bm{x}^\top P_t \bm{x} + z_t,
 \]
 which completes our inductive argument. 
 
  Lastly, by Lemma~\ref{lem:sol}, an optimal policy must be unique and it is obtained as \eqref{opt_policy}. \qed

 \section{Proof of Proposition~\ref{prop:mean}}\label{app:prop:mean}
 
 Recall that $g_t$ is given as~\eqref{g_eq}. 
 Let $\bar{g}_0 := g_0 (u_0^\star)$ and $\bar{g}_t := \mathbb{E} [g_t  (u_t^\star)]$ for $t \geq 1$, where $u_t^\star$ is an optimal control input.
 By \eqref{eq1} and \eqref{g_eq},
\[
\bar{x}_{t+1} = A\bar{x}_t - B R^{-1} B^\top \bar{g}_t + \frac{1}{\lambda} \Xi \Xi^\top \bar{g}_t.
\]
Furthermore, \eqref{eq_Q} implies that
\[
(P_t - Q)\bar{x}_t = A^\top \bar{g}_t. 
\]
Note that $\bar{g}_t = P_{t+1} \mathbb{E}\big [ Ax_t^\star + Bu_t^\star +\frac{1}{N} \sum_{i=1}^N \Xi  w_t^{\star, (i)} (u_t^\star) \big ]= P_{t+1} \bar{x}_{t+1}$. Combining the equations, the result follows. \qed

 \section{Proof of Lemma~\ref{lem:sol2}}\label{app:lem:sol2}
 
 Let $P$ be a solution to the equation $P-Q = A^\top P (I + WP)^{-1} A$. Let $E:=(I + WP)^{-1} A$ be decomposed as $E=U D U^{-1}$, where $D$ is a Jordan normal form.
Then, we have $P-Q = A^\top P U D U^{-1}$.
Let $V:=PU$. Then, we obtain 
\[
V-Q U = A^\top VD.
\]
Since $A = (I+WP)E=(I+WVU^{-1}) UDU^{-1}$, we have 
\[
AU = UD + WVD.
\]
Therefore, we obtain that  $F \begin{bmatrix} U \\ V \end{bmatrix}  = G\begin{bmatrix} U \\ V \end{bmatrix} D$. 
This implies that $P=VU^{-1}$ is a solution to the ARE~\eqref{are} and $\begin{bmatrix} U \\ V \end{bmatrix}$ solves generalized eigenvalue problem.
\qed

\section{Proof of Lemma~\ref{lem:stable}}\label{app:lem:stable}

Suppose first that $F \begin{bmatrix} \hat{V}_1 \\ \hat{V}_2 \end{bmatrix} = G \begin{bmatrix} \hat{V}_1 \\ \hat{V}_2 \end{bmatrix} \hat{M}$, where $\hat{M}$ is a Jordan normal form.
Then, $A = (I + W\hat{V}_2 \hat{V}_1^{-1}) \hat{V}_1 \hat{M} \hat{V}_1^{-1}$.
By the ARE~\eqref{are}, we have
\begin{equation*}
\begin{split}
    &\hat{V}_2 \hat{V}_1^{-1} = Q + A^\top \hat{V}_2 \hat{V}_1^{-1} (I+W\hat{V}_2 \hat{V}_1^{-1})^{-1} A\\
    &= Q + (\hat{V}_1^{-H} \hat{M}^H \hat{V}_1^H + \hat{V}_1^{-H} \hat{M}^H \hat{V}_2^H W^H ) \hat{V}_2 \hat{M} \hat{V}_1^{-1}.
\end{split}
\end{equation*}
This is a discrete-time Lyapunov equation of the form
\begin{equation}\label{lyap}
P = \bar{A}^H P \bar{A} + \bar{Q},
\end{equation}
 where $P=\hat{V}_2 \hat{V}_1^{-1}$, $\bar{A}:= \hat{V}_1 \hat{M} \hat{V}_1^{-1}$, and $\bar{Q} := Q+(\hat{V}_2 \hat{M} \hat{V}_1^{-1})^H W \hat{V}_2 \hat{M} \hat{V}_1^{-1}$.
Note that $\bar{Q} \succeq 0$ since $W \succeq 0$ under Assumption~\ref{ass:W}.
By the theory of Lyapunov equations, 
we conclude that $P \succeq 0$ since
$\bar{Q} \succeq 0$ and $\bar{A}$ is stable.

We  now  assume that $P \succeq 0$. Suppose that  $\bar{Q}  \succeq 0$, and $\bar{A}$ has an unstable eigenvalue, i.e., $\bar{A} v = \gamma v$, where $\vert \gamma \vert \geq 1$.
Pre-multiplying $v^H$ and post-multiplying $v$ on  both sides of the Lyapunov equation~\eqref{lyap}, we obtain $(\gamma^* \gamma -1) v^H P v + v^H \bar{Q} v = 0$.
Then, $\sqrt{\bar{Q}}v=0$, which leads to $\sqrt{Q} v =  \sqrt{W} \hat{V}_2 \hat{M} \hat{V}_1^{-1} v=0$ and $Av = (I+WP)\bar{A}v = \gamma v$.
This contradicts Assumption \ref{ass:ob}.
Therefore, if $P  \succeq 0$ and $\bar{Q}  \succeq 0$, then $\bar{A}$ must be stable.
Since $\bar{A}$ and $\hat{M}$ have the same spectrum, the result follows.
\qed

\bibliographystyle{IEEEtran}
\bibliography{reference} 

\end{document}